\documentclass{article}%
\usepackage{amsmath}
\usepackage{amsfonts}
\usepackage{amssymb}
\usepackage{graphicx}%
\providecommand{\U}[1]{\protect\rule{.1in}{.1in}}
\newtheorem{theorem}{Theorem}

\newtheorem{corollary}[theorem]{Corollary}

\newtheorem{lemma}[theorem]{Lemma}

\newenvironment{proof}[1][Proof]{\noindent\textbf{#1.} }{\ \rule{0.5em}{0.5em}}
\begin{document}

\author{Vadim E. Levit and David Tankus\\Department of Computer Science and Mathematics\\Ariel University, ISRAEL\\\{levitv, davidta\}@ariel.ac.il}
\title{Weighted Well-Covered Claw-Free Graphs}
\date{}
\maketitle

\begin{abstract}
A graph $G$ is \textit{well-covered} if all its maximal independent sets are
of the same cardinality. Assume that a weight function $w$ is defined on its
vertices. Then $G$ is $w$\textit{-well-covered} if all maximal independent
sets are of the same weight. For every graph $G$, the set of weight functions
$w$ such that $G$ is $w$-well-covered is a \textit{vector space}. Given an
input claw-free graph $G$, we present an $O\left(  n^{6}\right)  $ algortihm,
whose input is a claw-free graph $G$, and output is the vector space of weight
functions $w$, for which $G$ is $w$-well-covered.

A graph $G$ is \textit{equimatchable} if all its maximal matchings are of the
same cardinality. Assume that a weight function $w$ is defined on the edges of
$G$. Then $G$ is $w$\textit{-equimatchable} if all its maximal matchings are
of the same weight. For every graph $G$, the set of weight functions $w$ such
that $G$ is $w$-equimatchable is a vector space. We present an $O\left(
m\cdot n^{4}+n^{5}\log n\right)  $ algorithm which receives an input graph
$G$, and outputs the vector space of weight functions $w$ such that $G$ is $w$-equimatchable.

\end{abstract}

\section{Introduction}

\subsection{Basic Definitions and Notation}

Throughout this paper $G = (V,E)$ is a simple (i.e., a finite, undirected,
loopless and without multiple edges) graph with vertex set $V = V (G)$ and
edge set $E = E(G)$.

Cycles of $k$ vertices are denoted by $C_{k}$, and paths of $k$ vertices are
denoted by $P_{k}$. When we say that $G$ contains a $C_{k}$ or a $P_{k}$ for
some $k\geq3$, we mean that $G$ admits a subgraph isomorphic to $C_{k}$ or to
$P_{k}$, respectively. It is important to mention that these subgraphs are not
necessarily induced.

Let $u$ and $v$ be two vertices in $G$. The \textit{distance} between $u$ and
$v$, denoted $d(u,v)$, is the length of a shortest path between $u$ and $v$,
where the length of a path is the number of its edges. If $S$ is a non-empty
set of vertices, then the \textit{distance} between $u$ and $S$, denoted
$d(u,S)$, is defined by
\[
d(u,S)=\min\{d(u,s):s\in S\}.
\]
For every positive integer $i$, denote
\[
N_{i}(S)=\{x\in V:d(x,S)=i\},
\]
and
\[
N_{i}\left[  S\right]  =\{x\in V:d(x,S)\leq i\}.
\]

We abbreviate $N_{1}(S)$ and $N_{1}\left[  S\right]  $ to be $N(S)$ and
$N\left[  S\right]  $, respectively. If $S$ contains a single vertex, $v$,
then we abbreviate
\[
N_{i}(\{v\}),N_{i}\left[  \{v\}\right]  ,N(\{v\}),\text{ and }N\left[
\{v\}\right]
\]
to be
\[
N_{i}(v),N_{i}\left[  v\right]  ,N(v),\text{ and }N\left[  v\right]  ,
\]
respectively. We denote by $G[S]$ the subgraph of $G$ induced by $S$. For
every two sets, $S$ and $T$, of vertices of $G$, we say that $S$
\textit{dominates} $T$ if $T\subseteq N\left[  S\right]  $.

\subsection{Well-Covered Graphs}

Let $G=(V,E)$ be a graph. A set of vertices $S$ is \textit{independent} if its
elements are pairwise nonadjacent. An independent set of vertices is
\textit{maximal} if it is not a subset of another independent set. An
independent set of vertices is \textit{maximum} if the graph does not contain
an independent set of a higher cardinality.

The graph $G=(V,E)$ is \textit{well-covered} if every maximal independent set
is maximum. Assume that a weight function $w:V\longrightarrow\mathbb{R}$ is
defined on the vertices of $G$. For every set $S\subseteq V$, define
\[
w(S)=\sum_{s\in S}w(s).
\]
Then $G$ is $w$\textit{-well-covered} if all maximal independent sets of $G$
are of the same weight.

The problem of finding a maximum independent set in an input graph is
\textbf{NP-}complete. However, if the input is restricted to well-covered
graphs, then a maximum independent set can be found polynomially using the
\textit{greedy algorithm}. Similarly, if a weight function $w:V\longrightarrow
\mathbb{R}$ is defined on the vertices of $G$, and $G$ is $w$-well-covered,
then finding a maximum weight independent set is a polynomial problem.

The recognition of well-covered graphs is known to be \textbf{co-NP}-complete.
This was proved independently in \cite{cs:note} and \cite{sknryn:compwc}. In
\cite{cst:structures} it is proven that the problem remains \textbf{co-NP}%
-complete even when the input is restricted to $K_{1,4}$-free graphs. However,
the problem is polynomially solvable for $K_{1,3}$-free graphs
\cite{tata:wck13f,tata:wck13fn}, for graphs with girth at least $5$
\cite{fhn:wcg5}, for graphs with a bounded maximal degree \cite{cer:degree},
for chordal graphs \cite{ptv:chordal}, for bipartite graphs
\cite{favaron:verywell,plummer:survey,ravindra:well-covered}, and for graphs
without cycles of length $4$ and $5$ \cite{fhn:wc45}. It should be emphasized
that the forbidden cycles are not necessarily induced.

For every graph $G$, the set of weight functions $w$ for which $G$ is
$w$-well-covered is a \textit{vector space} \cite{cer:degree}. That vector
space is denoted $WCW(G)$ \cite{bnz:wcc4}. 

Clearly, $w\in WCW(G)$ if and only if $G$ is $w$-well-covered. Since
recognizing well-covered graphs is \textbf{co-NP}-complete, finding the vector
space $WCW(G)$ of an input graph $G$ is \textbf{co-NP}-hard. In
\cite{lt:wwc456} there is a polynomial algorithm which receives as its input a
graph $G$ without cycles of lengths $4$, $5$, and $6$, and outputs $WCW(G)$.

This article presents a polynomial algorithm whose input is a $K_{1,3}$-free
graph $G$, and the output is $WCW(G)$. Thus we generalize
\cite{tata:wck13f,tata:wck13fn}, which supply a polynomial time algorithm for
recognizing well-covered $K_{1,3}$-free graphs.

\subsection{Generating Subgraphs and Relating Edges}

We use the following notion, which has been introduced in \cite{lt:wc4567}.
Let $B$ be an induced complete bipartite subgraph of $G$ on vertex sets of
bipartition $B_{X}$ and $B_{Y}$. Assume that there exists an independent set
$S$ such that each of $S\cup B_{X}$ and $S\cup B_{Y}$ is a maximal independent
set of $G$. Then $B$ is a \textit{generating} subgraph of $G$, and it
\textit{produces} the restriction: $w(B_{X})=w(B_{Y})$. Every weight function
$w$ such that $G$ is $w$-well-covered must \textit{satisfy} the restriction
$w(B_{X})=w(B_{Y})$. The set $S$ is a \textit{witness} that $B$ is generating.
In the restricted case that the generating subgraph $B$ is isomorphic to
$K_{1,1}$, call its vertices $x$ and $y$. In that case $xy$ is a
\textit{relating} edge, and $w(x)=w(y)$ for every weight function $w$ such
that $G$ is $w$-well-covered.

The decision problem whether an edge in an input graph is relating is
\textbf{NP-}complete \cite{bnz:wcc4}. Therefore, recognizing generating
subgraphs is \textbf{NP-}complete as well. In \cite{lt:relatedc4} it is proved
that recognizing relating edges and generating subgraphs is \textbf{NP-}%
complete even in graphs without cycles of lengths $4$ and $5$. However,
recognizing relating edges can be done polynomially if the input is restricted
to graphs without cycles of lengths $4$ and $6$ \cite{lt:relating}, and
recognizing generating subgraphs is a polynomial problem when the input is
restricted to graphs without cycles of lengths $4$, $6$ and $7$
\cite{lt:wc4567}.

Generating subgraphs play an important roll in finding the vector space
$WCW(G)$. In this article we use generating subgraphs in the algorithm which
receives as its input a $K_{1,3}$-free graph $G$, and outputs $WCW(G)$.

\subsection{Equimatchable Graphs}

Let $G=(V,E)$ be a graph. The \textit{line graph} of $G$, denoted $L(G)$ is a
graph such that every vertex of $L(G)$ represents an edge in $G$, and two
vertices of $L(G)$ are adjacent if and only if they represent two edges in $G$
with a common endpoint.

Every independent set of vertices in $L(G)$ defines a set of pairwise
non-adjacent edges in $G$. A set of pairwise non-adjacent edges is called a
\textit{matching}. A matching $M$ \textit{dominates} a set $S$ of vertices if
every vertex of $S$ is an endpoint of an edge of $M$. A matching in a graph is
\textit{maximal} if it is not contained in another matching. 

The \textit{size} of a matching $M$, denoted $|M|$, is the number of its
edges. A matching $M$ is \textit{maximum} if the graph does not admit a
matching with size bigger than $|M|$. 

A graph is called \textit{equimatchable} if all its maximal matchings are
maximum. Clearly, $G$ is equimatchable if and only if $L(G)$ is well-covered.

Line graphs are characterized by a list of forbidden induced subgraphs
\cite{lpp:equimatchable}. One of these subgraphs is $K_{1,3}$, called a
\textit{claw}. 

Hence, every line graph is claw-free. Thus the existence of a polynomial
algorithm for recognizing well-covered claw-free graphs
\cite{tata:wck13f,tata:wck13fn}, implies a polynomial algorithm for
recognizing equimatchable graphs.

Assume that a weight function $w:E\longrightarrow\mathbb{R}$ is defined on the
edges of $G$. For every set $S\subseteq E$, define
\[
w(S)=\sum_{s\in S}w(s).
\]
Then $G$ is $w$-\textit{equimatchable} if all its maximal matchings are of the
same weight.

It is easy to see that for every graph $G$, the set of weight functions $w$
such that $G$ is $w$-equimatchable is a vector space. We denote that vector
space by $EVS(G)$.

In this paper we present a polynomial algorithm whose input is a graph $G$,
and the output is the vector space $EVS(G)$.

\section{Weighted Hereditary Systems}

A \textit{hereditary system} is a pair $H=(S,\Psi)$, where $S$ is a finite set
and $\Psi$ is a family of subsets of $S$, where $f\in\Psi$ and $f^{\prime
}\subseteq f$ implies $f^{\prime}\in\Psi$. The members of $\Psi$ are called
the \textit{feasible} sets of the system. 

A feasible set is \textit{maximal} if it is not contained in another feasible
set. A feasible set is \textit{maximum} if the hereditary system does not
admit a feasible set with higher cardinality.

A hereditary system is \textit{greedy} if and only if its maximal feasible
sets are all of the same cardinality. Equivalently, a greedy hereditary system
is a hereditary system for which the \textit{greedy algorithm} for finding a
maximal feasible set always produces a maximum cardinality feasible set.

Assume that a weight function $w:S\longrightarrow\mathbb{R}$ is defined on the
elements of a hereditary system. The hereditary system is \textit{greedy} if
and only if all its maximal feasible sets are of the same weight, and
equivalently, the greedy algorithm for finding a maximal feasible set always
produces a feasible set of maximum weight.

An example of the above is a hereditary system $H=(S,\Psi)$, where $S=V$ is
the set of vertices of a given graph $G=(V,E)$, and $\Psi$ is the family of
all independent sets of $G$. Clearly, the hereditary system $H=(V,\Psi)$ is
greedy if and only if $G$ is well-covered. Similarly, if a weight function
$w:V\longrightarrow\mathbb{R}$ is defined, then the hereditary system
$H=(V,\Psi)$ is greedy if and only if $G$ is $w$-well-covered.

Another example of a hereditary system is a pair $H=(S,\Psi)$, where $S=E$ is
the set of edges of a graph $G=(V,E)$, and $\Psi$ is the family of its
matchings. Clearly, the hereditary system $H=(E,\Psi)$ is greedy if and only
if the graph is equimatchable. Similarly, if a weight function
$w:E\longrightarrow\mathbb{R}$ is defined, then the hereditary system
$H=(E,\Psi)$ is greedy if and only if $G$ is $w$-equimatchable.

\begin{theorem}
\cite{tata:hamilton} \label{hs} Let $H=(S,\Psi)$ be a hereditary system. Then
$H$ is not greedy if and only if there exist two maximal feasible sets,
$F_{1}$ and $F_{2}$, of $S$ with different cardinalities, \ $\left\vert
F_{1}\right\vert \not =\left\vert F_{2}\right\vert $, \ such that for each
$f_{1}\in F_{1}\setminus F_{2}$, and for each $f_{2}\in F_{2}\setminus F_{1}$,
the set
\[
(F_{1}\cap F_{2})\cup\{f_{1},f_{2}\}
\]
is not feasible.
\end{theorem}

The following is a generalization of Theorem \ref{hs}.

\begin{theorem}
\label{whs} Let
\[
(H=(S,\Psi),w:S\longrightarrow\mathbb{R})
\]
be a hereditary system with a weight function defined on its elements. Then
$(H,w)$ is not greedy if and only if there exist two maximal feasible sets,
$F_{1}$ and $F_{2}$, of $S$ with different weights, $w(F_{1})\not =w(F_{2})$,
such that for each $f_{1}\in F_{1}\setminus F_{2}$, and for each $f_{2}\in
F_{2}\setminus F_{1}$, the set
\[
(F_{1}\cap F_{2})\cup\{f_{1},f_{2}\}
\]
is not feasible.
\end{theorem}

\begin{proof}
Clearly, if there exist two maximal feasible sets with different weights then
the hereditary system is not greedy.

Suppose $(H,w)$ is not greedy. There exist two maximal feasible sets, $F_{1}$
and $F_{2}$, of $S$ with the following two properties:

\begin{enumerate}
\item $w(F_{1}) \not = w(F_{2})$.

\item For every two maximal feasible sets, $F_{1}^{\prime}$ and $F_{2}%
^{\prime}$, of $S$, if $w(F_{1}^{\prime})\not =w(F_{2}^{\prime})$ then
\[
\left\vert F_{1}\cap F_{2}\right\vert \geq|F_{1}^{\prime}\cap F_{2}^{\prime}|.
\]

\end{enumerate}

Assume \ on \ the \ contrary \ that \ there \ \ exist \ \ $f_{1}\in
F_{1}\setminus F_{2}$, \ \ and \ \ $f_{2}\in F_{2}\setminus F_{1}$, \ \ such
\ \ that \ \ the \ set \
\[
F_{3}=(F_{1}\cap F_{2})\cup\{f_{1},f_{2}\}
\]
\ is \ feasible. \ Clearly, \
\[
\min\{|F_{1}\cap F_{3}|,|F_{2}\cap F_{3}|\}>|F_{1}\cap F_{2}|.
\]
\ \ Therefore, \ \ $w(F_{1})=w(F_{3})$ \ \ and \ \ $w(F_{2})=w(F_{3})$. Hence,
$w(F_{1})=w(F_{2})$, which is a contradiction.

We proved that for every $f_{1}\in F_{1}\setminus F_{2}$, and for every
$f_{2}\in F_{2}\setminus F_{1}$, the set
\[
(F_{1}\cap F_{2})\cup\{f_{1},f_{2}\}
\]
is not feasible.
\end{proof}

\section{$w$-Well-Covered Claw-Free Graphs}

The following is an instance of Theorem \ref{whs}:

\begin{theorem}
\label{whswc} Let
\[
(G=(V,E),w:V\longrightarrow\mathbb{R})
\]
be a graph with a weight function defined on its vertices. Then $G$ is not
$w$-well-covered if and only if there exist two maximal independent sets,
$S_{1}$ and $S_{2}$, with different weights, $w(S_{1})\not =w(S_{2})$, such
that the subgraph induced by $S_{1}\bigtriangleup S_{2}$ is complete bipartite.
\end{theorem}

\begin{corollary}
\label{generating} Let $G=(V,E)$ be a graph, and let $B$ be an induced
complete bipartite subgraph of $G$ on vertex sets of bipartition $B_{X}$ and
$B_{Y}$. Then the following two conditions are equivalent:

\begin{enumerate}
\item There exist two maximal independent sets, $S_{1}$ and $S_{2}$, of $G$
such that $B_{X} = S_{1} \setminus S_{2}$ and $B_{Y} = S_{2} \setminus S_{1}$.

\item $B$ is generating.
\end{enumerate}
\end{corollary}

\begin{proof}
If the first condition holds then $S_{1} \cap S_{2}$ is a witness that $B$ is generating.

If $B$ is generating, let $S$ be a witness of $B$. The first condition holds
for $S_{1} = S \cup B_{X}$ and $S_{2} = S \cup B_{Y}$.
\end{proof}

The main result of this section is the following.

\begin{theorem}
\label{wcwclaw} There exists an $O(|V|^{6})$ algorithm, which receives as its
input a claw-free graph $G$, and finds $WCW(G)$.
\end{theorem}

\begin{proof}
Let $G =(V, E)$ be a graph. The following algorithm finds $WCW(G)$.

\begin{enumerate}
\item For every induced complete bipartite subgraph $B$ of $G$

\begin{enumerate}
\item Denote its vertex sets of bipartition $B_{X}$ and $B_{Y}$.

\item Decide whether $B$ is generating.

\item If $B$ is generating

\begin{enumerate}
\item List the restriction $w(B_{X})=w(B_{Y})$.
\end{enumerate}
\end{enumerate}

\item $w \in WCW(G)$ if and only if $w$ satisfies all listed restrictions.
\end{enumerate}

In the general case, this algorithm is not polynomial, because the number of
induced complete bipartite subgraphs is not polynomial, and the time needed to
decide whether one induced subgraph is generating is not polynomial. However,
we show that the algorithm can be implemented polynomially if the input graph
is claw-free.

Suppose $G$ is claw-free. Then every induced complete bipartite subgraph is
isomorphic to one of the following graphs: $K_{1,1}$, $K_{1,2}$, and $K_{2,2}%
$. Hence, the number of subgraphs the algorithm checks is polynomial. It
remains to prove that it is possible to decide polynomially for a single
subgraph whether it is generating.

Let $B$ be an induced complete bipartite subgraph of $G$ on vertex sets of
bipartition $B_{X}$ and $B_{Y}$. Define
\[
M_{1}=(N(B_{X})\cap N_{2}(B_{Y}))\cup(N_{2}(B_{X})\cap N(B_{Y})),
\]
and
\[
M_{2}=(N_{2}(B_{X})\cap N_{3}(B_{Y}))\cup(N_{3}(B_{X})\cap N_{2}(B_{Y})).
\]
Clearly, $B$ is generating if and only if there exists an independent set in
$M_{2}$ which dominates $M_{1}$.

If $B = K_{2,2}$ then the fact that the graph is claw-free implies that
$M_{1}$ and $M_{2}$ are empty sets. Hence, $M_{2}$ dominates $M_{1}$, and $B$
is generating.

Assume $B\not =K_{2,2}$. In order to decide whether $B$ is generating, define
a weight function
\[
w:M_{2}\longrightarrow\mathbb{R}\text{ by }w(x)=|N(x)\cap M_{1}|,
\]
i.e. the weight of every vertex in $M_{2}$ is the number of vertices it
dominates in $M_{1}$. The fact that the graph is claw-free implies that a
vertex of $M_{1}$ can not be dominated by two non-adjacent vertices of $M_{2}%
$. Therefore, if $S\subseteq M_{2}$ is independent then
\[
w(S)=\sum_{s\in S}w(s)=\sum_{s\in S}\left\vert N(s)\cap M_{1}\right\vert
=|N(S)\cap M_{1}|,
\]
i.e., the weight of $S$ is the number of vertices it dominates in $M_{1}$.

The next step is to invoke an algorithm finding the maximum weighted
independent set in claw-free graphs. First such algorithm is due to Minty
\cite{minty:wclaw}, while the best known one with the complexity $O(|V|^{3})$
may be found in \cite{faenza2011}. Let $S^{\ast}$ be a maximum weight
independent set of $G[M_{2}]$. Clearly, $w(S^{\ast})\leq|M_{1}|$. If
$w(S^{\ast})=|M_{1}|$ then $S^{\ast}$ dominates $M_{1}$, and $B$ is
generating. Otherwise, there does not exist an independent set of $M_{2}$
which dominates $M_{1}$, and $B$ is not generating.

The number of induced complete bipartite subgraphs which are isomorphic to
$K_{1,1}$ or $K_{1,2}$ is $O(|V|^{3})$. Hence, the complexity of the algorithm
is $O(|V|^{6})$.
\end{proof}

\section{$w$-Equimatchable Graphs}

Let $G=(V,E)$ be a graph and $w:E\longrightarrow\mathbb{R}$ a weight function
defined on its vertices. Since there is a 1-1 mapping between the edges of $G$
and the vertices of $L(G)$, the function $w$ can be viewed as a weight
function on the vertices of $L(G)$. Therefore, $G$ is $w$-equimatchable if and
only if $L(G)$ is $w$-well-covered. Hence, $EVS(G)=WCW(L(G))$. Obviously,
$EVS(G)$ can be found polynomially by constructing the line-graph, $L(G)$, and
then applying the algorithm of the proof of Theorem \ref{wcwclaw} to find
$WCW(L(G))$. The number of vertices in $L(G)$ is $|E|$. Hence, $WCW(L(G))$ can
be found in $|E|^{6}$ time. However, the main result of this section is an
algorithm which finds $EVS(G)$ in $O\left(  \left\vert E\right\vert
\cdot\left\vert V\right\vert ^{4}+\left\vert V\right\vert ^{5}\log\left\vert
V\right\vert \right)  $ time.

The following is an instance of Theorem \ref{whs}.

\begin{theorem}
\label{whseq} Let
\[
(G=(V,E),w:E\longrightarrow\mathbb{R})
\]
be a graph with a weight function defined on its edges. Then $G$ is not
$w$-equimatchable if and only if there exist two maximal matchings, $M_{1}$
and $M_{2}$, with different weights, $w(M_{1})\not =w(M_{2})$, such that
$M_{1}\bigtriangleup M_{2}$ is one of the following.

\begin{enumerate}
\item Two adjacent edges, $v_{1}v_{2} \in M_{1}$ and $v_{2}v_{3} \in M_{2}$.

\item Three edges, $\{v_{1}v_{2}, v_{3}v_{4}\} \subseteq M_{1}$ and
$v_{2}v_{3} \in M_{2}$.

\item Four edges, $\{v_{1}v_{2}, v_{3}v_{4}\} \subseteq M_{1}$, and
$\{v_{2}v_{3}, v_{1}v_{4}\} \subseteq M_{2}$.
\end{enumerate}
\end{theorem}

We need the following three lemmas to prove the main result of this section.

\begin{lemma}
\label{path} The following problem can be solved in $O\left(  \left\vert
E\right\vert \cdot\left\vert V\right\vert +\left\vert V\right\vert ^{2}%
\log\left\vert V\right\vert \right)  $ time.\newline\textit{Instance:} A graph
$G=(V,E)$ and a path $P=(v_{1}v_{2},v_{2}v_{3},...,v_{k-1}v_{k})$ in $G$ for
some $k\geq3$.\newline\textit{Question:} Do there exist two maximal matchings,
$M_{1}$ and $M_{2}$, of $G$ such that $P=M_{1}\bigtriangleup M_{2}$?
\end{lemma}

\begin{proof}
If $k$ is even and $v_{1}v_{k} \in E$ then the instance is obviously negative.
Hence, we assume that $k$ is odd or $v_{1}v_{k} \not \in E$.

Define
\[
V^{\prime}=V\setminus\{v_{1},...,v_{k}\}\text{ and }D=N(\{v_{1},v_{k}\})\cap
V^{\prime}.
\]
Let $G^{\prime}$ be the induced subgraph of $G$ on vertex set $V^{\prime}$,
and denote the set of its edges by $E^{\prime}$.

Define a weight function $w:E^{\prime}\longrightarrow\mathbb{R}$ by:
\[
\forall xy\in E^{\prime}\ \ w(xy)=|\{x,y\}\cap D|.
\]
For every matching $M$ in $G^{\prime}$, its weight, $w(M)$, equals to the
number of vertices of $D$ which are dominated by $M$. We now invoke the
algorithm of \cite{gabow:wmatching} for finding a maximum weight matching in a
graph, and denote the output of the algorithm by $M^{\ast}$. Clearly,
$w(M^{\ast})\leq|D|$.

Suppose $w(M^{\ast})=|D|$. Then $M^{\ast}$ dominates $D$. Let $M^{\ast\ast}$
be any maximal matching in $G^{\prime}$ which contains $M^{\ast}$, and define
\begin{align*}
M_{1} &  =M^{\ast\ast}\cup\{v_{2i-1}v_{2i}:1\leq2i\leq k\},\\
M_{2} &  =M^{\ast\ast}\cup\{v_{2i}v_{2i+1}:3\leq2i+1\leq k\}.
\end{align*}
Obviously, $M_{1}$ and $M_{2}$ are two maximal matchings of $G$ and
$P=M_{1}\bigtriangleup M_{2}$.

On the other hand, suppose $w(M^{\ast})<\left\vert D\right\vert $. There does
not exist a maximal matching of $G^{\prime}$ which dominates $D$, and
therefore the instance at hand is negative.

The complexity of the algorithm of \cite{gabow:wmatching} is $O\left(
\left\vert E\right\vert \cdot\left\vert V\right\vert +\left\vert V\right\vert
^{2}\log\left\vert V\right\vert \right)  $. This is also the complexity of
this algorithm.
\end{proof}

\begin{lemma}
\label{cycle} Let $G=(V,E)$ be a graph, and let $C=(v_{1}v_{2},v_{2}%
v_{3},...,v_{k-1}v_{k},v_{k}v_{1})$ be an even cycle in $G$, for some $k\geq
4$. Then there exist two maximal matchings, $M_{1}$ and $M_{2}$, of $G$ such
that $C=M_{1}\bigtriangleup M_{2}$.
\end{lemma}

\begin{proof}
Let $M$ be any maximal matching in $G[V\setminus\{v_{1},...,v_{k}\}]$. Define
\[
M_{1}=M\cup\{v_{2i-1}v_{2i}:1\leq i\leq\frac{k}{2}\}
\]
and
\[
M_{2}=M\cup\{v_{2i}v_{2i+1}:1\leq i\leq\frac{k}{2}-1\}\cup\{v_{k}v_{1}\}.
\]
Obviously, $M_{1}$ and $M_{2}$ are two maximal matchings of $G$ and
$C=M_{1}\bigtriangleup M_{2}$.
\end{proof}

The naive algorithm for finding $EVS(G)$ checks all structures described in
Theorem \ref{whseq}, i.e., all paths of lengths $2$ and $3$, and cycles of
length $4$. For each of these structures, the algorithm decides whether it is
the symmetric difference of two maximal matchings. If so, an appropriate
equation is added to the list of restrictions. A weight function
$w:E\longrightarrow\mathbb{R}$ satisfies all the restrictions found by the
algorithm if and only if $w\in EVS(G)$. For each path of lengths $2$ or $3$,
the naive algorithm invokes the algorithm of Lemma \ref{path}. By Lemma
\ref{cycle}, every cycle of length $4$ is the symmetric difference of two
maximal matchings. Hence, the total complexity of the naive algorithm is
\[
O\left(  \left\vert E\right\vert \cdot\left\vert V\right\vert ^{5}+\left\vert
V\right\vert ^{6}\log\left\vert V\right\vert \right)  .
\]
However, we present a more efficient algorithm.

\begin{lemma}
\label{allendpointsp4} The following problem can be solved in $O\left(
\left\vert E\right\vert \cdot\left\vert V\right\vert ^{2}+\left\vert
V\right\vert ^{3}\log\left\vert V\right\vert \right)  $ time:\newline Input: A
graph $G=(V,E)$, and two non-adjacent vertices, $v_{1}$ and $v_{4}$, in
$G$.\newline Output: All paths $P=(v_{1}v_{2},v_{2}v_{3},v_{3}v_{4})$, such
that there exist two maximal matchings, $M_{1}$ and $M_{2}$, in $G$, and
$M_{1}\bigtriangleup M_{2}=P$?
\end{lemma}

\begin{proof}
Let $G^{\prime}=(V^{\prime},E^{\prime})$ be the subgraph of $G$ induced by
$V^{\prime}=V\setminus\{v_{1},v_{4}\}$, and let $\epsilon=\frac{1}{|V|}$.
Define $w:E^{\prime}\longrightarrow\mathbb{R}$ as follows:
\[
w(xy)=\left\{
\begin{array}
[c]{ll}%
2+\epsilon & \quad\text{if $x\in N(v_{1})$ and $y\in N(v_{4})$ }\\
|\{x,y\}\cap N(\{v_{1},v_{4}\})| & \quad\text{otherwise }%
\end{array}
\right.
\]
(see Figure \ref{fw}).

\begin{figure}[h]
\setlength{\unitlength}{1.0cm} \begin{picture}(20,5)\thicklines
\put(4,3){\circle*{0.2}}
\put(4.5,3){\makebox(0,0){$v_{1}$}}
\put(5.5,4){\circle*{0.2}}
\put(5.5,2){\circle*{0.2}}
\put(7,4){\circle*{0.2}}
\put(7,2){\circle*{0.2}}
\put(8.5,3){\circle*{0.2}}
\put(8,3){\makebox(0,0){$v_{4}$}}
\put(4,3){\line(3,2){1.5}}
\put(4,3){\line(3,-2){1.5}}
\put(5.5,4){\line(1,0){1.5}}
\put(5.5,4){\line(4,-5){1.5}}
\put(5.5,2){\line(1,0){1.5}}
\put(6.2,4.2){\makebox(0,0){$2+\epsilon$}}
\put(6.2,3.2){\makebox(0,0){$2+\epsilon$}}
\put(6.2,2.2){\makebox(0,0){$2+\epsilon$}}
\put(7,2){\line(3,2){1.5}}
\put(7,4){\line(3,-2){1.5}}
\put(4,5){\circle*{0.2}}
\put(2.5,4){\circle*{0.2}}
\put(2.5,3){\circle*{0.2}}
\put(2.5,2){\circle*{0.2}}
\put(4,1){\circle*{0.2}}
\put(4,3){\line(0,1){2}}
\put(4,3){\line(-3,2){1.5}}
\put(4,3){\line(-1,0){1.5}}
\put(4,3){\line(-3,-2){1.5}}
\put(4,3){\line(0,-1){2}}
\multiput(1,1)(0,2){3}{\circle*{0.2}}
\put(1,1){\line(0,1){2}}
\put(1.2,2){\makebox(0,0){$0$}}
\put(1,1){\line(1,0){3}}
\put(2.5,1.2){\makebox(0,0){$1$}}
\put(2.5,5.2){\makebox(0,0){$1$}}
\put(1,3){\line(3,2){1.5}}
\put(2.5,3){\line(0,1){1}}
\put(1,3){\line(3,-2){3}}
\put(1,3){\line(1,0){1.5}}
\put(2,2.6){\makebox(0,0){$1$}}
\put(2,3.2){\makebox(0,0){$1$}}
\put(2,4){\makebox(0,0){$1$}}
\put(3.2,1.8){\makebox(0,0){$2$}}
\put(1,5){\line(1,0){3}}
\put(2.7,3.4){\makebox(0,0){$2$}}
\put(0,3){\circle*{0.2}}
\put(0,3){\line(1,2){1}}
\put(0,3){\line(1,-2){1}}
\put(0.5,4.3){\makebox(0,0){$0$}}
\put(0.5,2.3){\makebox(0,0){$0$}}
\multiput(10,1)(0,2){3}{\circle*{0.2}}
\put(8.5,3){\line(3,4){1.5}}
\put(8.5,3){\line(1,0){1.5}}
\put(8.5,3){\line(3,-4){1.5}}
\multiput(11.5,1)(0,2){3}{\circle*{0.2}}
\put(10,1){\line(1,0){1.5}}
\put(10.7,1.2){\makebox(0,0){$1$}}
\put(10,3){\line(0,1){2}}
\put(10,3){\line(3,-4){1.5}}
\put(10.7,2.4){\makebox(0,0){$1$}}
\put(10,3){\line(1,0){1.5}}
\put(10.7,3.2){\makebox(0,0){$1$}}
\put(10,5){\line(1,0){1.5}}
\put(10.7,5.2){\makebox(0,0){$1$}}
\put(11.5,5){\line(0,-1){2}}
\put(10.2,4){\makebox(0,0){$2$}}
\put(11.7,4){\makebox(0,0){$0$}}
\end{picture}\caption{The weight function $w$. Note that $w$ is not defined on
edges which dominate $v_{1}$ and $v_{4}$.}%
\label{fw}%
\end{figure}
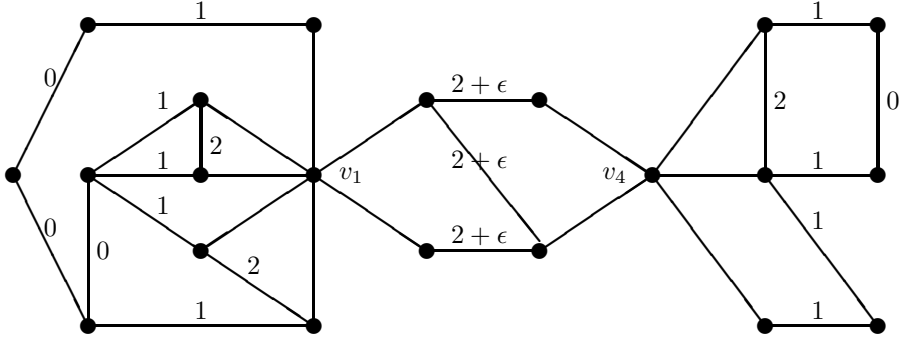

For every matching $M$ in $G^{\prime}$. There exist two integers, $0 \leq A
\leq|V|$ and $0 \leq B \leq|V|$, such that $w(M) = A + B\epsilon$, where $A$
is the number of vertices of $N(\{ v_{1}, v_{4} \})$ dominated by $M$, and $B$
is the number of edges of $M$ with one endpoint in $N(v_{1})$ and another
endpoint in $N(v_{4})$.

Let $M^{*}$ be a maximum weight matching in $G^{\prime}$. Let $0 \leq A
\leq|V|$ and $0 \leq B \leq|V|$ such that $w(M^{*}) = A + B \epsilon$. Among
all maximal matchings in $G^{\prime}$, the matching $M^{*}$ dominates maximum
possible number of vertices in $N(\{ v_{1}, v_{4} \})$. Among all maximal
matchings in $G^{\prime}$, which dominate $A$ vertices of $N(\{ v_{1}, v_{4}
\})$, the matching $M^{*}$ contains maximum number of edges with one endpoint
in $N(v_{1})$ and another endpoint in $N(v_{4})$.

Clearly, $A\leq|N(\{v_{1},v_{4}\})|$. If $A=|N(\{v_{1},v_{4}\})|$ and $B>0$
then $M^{\ast}$ dominates $N(\{v_{1},v_{4}\})$, and contains at least one edge
$v_{2}v_{3}$ where $v_{2}\in N(v_{1})$ and $v_{3}\in N(v_{4})$. Hence,
$M^{\ast}$ and
\[
M^{\ast\ast}=(M^{\ast}\cup\{v_{1}v_{2},v_{3}v_{4}\})\setminus\{v_{2}v_{3}\}
\]
are two maximal matchings in $G$, and $M^{\ast}\bigtriangleup M^{\ast\ast}=P$.

If $A=|N(\{v_{1},v_{4}\})|$ and $B=0$ then there exist matchings of
$G^{\prime}$ which dominate $N(\{v_{1},v_{4}\})$, but non of them contains an
edge $v_{2}v_{3}$ such that $v_{2}\in N(v_{1})$ and $v_{3}\in N(v_{4})$.
Therefore, there does not exist a path $(v_{1}v_{2},v_{2}v_{3},v_{3}v_{4})$,
which is the symmetric difference of two maximal matchings.

If $A<|N(\{v_{1},v_{4}\})|$ then there does not exist a matching of
$G^{\prime}$ which dominates $N(\{v_{1},v_{4}\})$, and therefore there does
not exist a path $(v_{1}v_{2},v_{2}v_{3},v_{3}v_{4})$ which is the symmetric
difference of two maximal matchings.

The following algorithm solves the problem.

\begin{enumerate}
\item Define $G^{\prime}$, $\epsilon$ and $w$ as above.

\item Invoke the algorithm of \cite{gabow:wmatching} to find a maximum weight
matching $M^{\prime}$ in $G^{\prime}$.

\item While $w(M^{\prime})>|N(\{v_{1},v_{4}\})|$

\begin{enumerate}
\item For every edge $v_{2}v_{3}\in M^{\prime}$ such that $w(v_{2}%
v_{3})=2+\epsilon$

\begin{enumerate}
\item List the path $( v_{1}, v_{2}, v_{3}, v_{4} )$.

\item Set $w(v_{2}v_{3}) = 2$.
\end{enumerate}

\item Invoke again the algorithm of \cite{gabow:wmatching} with the modified
definition of $w$, and get a new maximum weight matching $M^{\prime}$ in
$G^{\prime}$.
\end{enumerate}
\end{enumerate}

The complexity of the algorithm of \cite{gabow:wmatching} is $O\left(
\left\vert E\right\vert \cdot\left\vert V\right\vert +\left\vert V\right\vert
^{2}\log\left\vert V\right\vert \right)  $, and it is invoked at most $O(|V|)$
times. Hence, the total complexity of this algorithm is $O\left(  \left\vert
E\right\vert \cdot\left\vert V\right\vert ^{2}+\left\vert V\right\vert
^{3}\log\left\vert V\right\vert \right)  $. Note that if $v_{1}$ and $v_{4}$
are not the endpoints of a path of length $3$, which is the symmetric
difference of two maximal matchings, then the algorithm of
\cite{gabow:wmatching} is invoked only once. In this restricted case the
complexity of the algorithm is $O\left(  \left\vert E\right\vert
\cdot\left\vert V\right\vert +\left\vert V\right\vert ^{2}\log\left\vert
V\right\vert \right)  $.
\end{proof}

The next theorem is the main result of this section.

\begin{theorem}
\label{fastevs} The following problem can be solved in $O\left(  \left\vert
E\right\vert \cdot\left\vert V\right\vert ^{4}+\left\vert V\right\vert
^{5}\log\left\vert V\right\vert \right)  $ time:\newline Input: A graph
$G=(V,E)$.\newline Output: $EVS(G)$.
\end{theorem}

\begin{proof}
The following algorithm solves the problem:

\begin{enumerate}
\item For each subgraph $H$ (not necessarily induced) isomorphic to $P_{3}$ on
vertex set $\{v_{1}, v_{2}, v_{3}\}$:

\begin{enumerate}
\item Invoke the algorithm described in the proof of Lemma \ref{path} to
decide whether $\left(  v_{1}v_{2},v_{2}v_{3}\right)  $ is the symmetric
difference of two maximal matchings.

\item If so, add the restriction: $w(v_{1}v_{2})=w(v_{2}v_{3})$.
\end{enumerate}

\item For each pair of non-adjacent vertices, $v_{1}$ and $v_{4}$:

\begin{enumerate}
\item Invoke the algorithm of Lemma \ref{allendpointsp4}.

\item For each path $(v_{1}v_{2},v_{2}v_{3},v_{3}v_{4})$ found by the algorithm:

\begin{enumerate}
\item Add the restriction: $w(v_{1}v_{2})+w(v_{3}v_{4})=w(v_{2}v_{3})$.
\end{enumerate}
\end{enumerate}

\item For each subgraph (not necessarily induced) isomorphic to $C_{4}$ on
vertex set $\{v_{1},v_{2},v_{3},v_{4}\}$:

\begin{enumerate}
\item Add the restriction: $w(v_{1}v_{2})+w(v_{3}v_{4})=w(v_{2}v_{3}%
)+w(v_{1}v_{4})$.
\end{enumerate}
\end{enumerate}

The complexity of the algorithm of Lemma \ref{path} is $O\left(  \left\vert
E\right\vert \cdot\left\vert V\right\vert +\left\vert V\right\vert ^{2}%
\log\left\vert V\right\vert \right)  $, and it is invoked $O(\left\vert
V\right\vert ^{3})$ times in step 1. Hence, the complexity of step 1 is
$O\left(  \left\vert E\right\vert \cdot\left\vert V\right\vert ^{4}+\left\vert
V\right\vert ^{5}\log\left\vert V\right\vert \right)  $. The complexity of the
algorithm of Lemma \ref{allendpointsp4} is $O\left(  \left\vert E\right\vert
\cdot\left\vert V\right\vert ^{2}+\left\vert V\right\vert ^{3}\log\left\vert
V\right\vert \right)  $, and it is invoked $O(|V|^{2})$ times in step 2.
Hence, the complexity of step 2 is $O\left(  \left\vert E\right\vert
\cdot\left\vert V\right\vert ^{4}+\left\vert V\right\vert ^{5}\log\left\vert
V\right\vert \right)  $. The complexity of step 3 is $O(|V|^{4})$. Thus the
total complexity of this algorithm is
\[
O\left(  \left\vert E\right\vert \cdot\left\vert V\right\vert ^{4}+\left\vert
V\right\vert ^{5}\log\left\vert V\right\vert \right)  .
\]

\end{proof}

\section{Conclusion and Future Work}

A graph $G$ is equimatchable if and only if $EVS(G)$ contains the function
$w\equiv1$. It follows from Theorem \ref{whseq} that $G$ is equimatchable if
and only if there do not exist two maximal matchings, $M_{1}$ and $M_{2}$,
such that $M_{1}\bigtriangleup M_{2}$ is a path of length $3$.

Hence, the following algorithm decides whether $G$ is equimatchable: For every
pair of non-adjacent vertices, $v_{1}$ and $v_{4}$, in $G$, invoke the
algorithm of Lemma \ref{allendpointsp4} with input $\left(  G,v_{1}%
,v_{4}\right)  $. Once the algorithm of Lemma \ref{allendpointsp4} yields a
non-empty list of paths, this algorithm outputs that $G$ is not equimatchable.
If all calls of the algorithm of Lemma \ref{allendpointsp4} yielded empty
lists of paths, then $G$ is equimatchable.

The algorithm of Lemma \ref{allendpointsp4} is called at most $O(|V|^{2})$
times. However, all of these calls, except maybe the last one, yielded empty
lists. The complexity invoking the algorithm of Lemma \ref{allendpointsp4} and
receiving an empty output is $O\left(  \left\vert E\right\vert \cdot\left\vert
V\right\vert +\left\vert V\right\vert ^{2}\log\left\vert V\right\vert \right)
$, while the complexity invoking the algorithm of Lemma \ref{allendpointsp4}
and receiving a non empty output is $O\left(  \left\vert E\right\vert
\cdot\left\vert V\right\vert ^{2}+\left\vert V\right\vert ^{3}\log\left\vert
V\right\vert \right)  $. Hence, the total complexity of this algorithm is
$O\left(  \left\vert E\right\vert \cdot\left\vert V\right\vert ^{3}+\left\vert
V\right\vert ^{4}\log\left\vert V\right\vert \right)  $ time.

However, for this restricted case a more efficient algorithm has been found in
\cite{de:efficienteq}. That algorithm decides whether an input graph is
equimatchable in $O(\left\vert E\right\vert \cdot\left\vert V\right\vert
^{2})$ time. It seems worth trying to improve on our algorithm returning
$EVS(G)$ using the technique presented in \cite{de:efficienteq}.


\begin{thebibliography}{99}                                                                                               %


\bibitem {bnz:wcc4}J. I. Brown, R. J. Nowakowski, I. E. Zverovich, \emph{The
structure of well-covered graphs with no cycles of length 4}, Discrete
Mathematics \textbf{307} (2007) 2235-2245.

\bibitem {cer:degree}Y. Caro, N. Ellingham, G. F. Ramey, \emph{Local structure
when all maximal independent sets have equal weight}, SIAM Journal on Discrete
Mathematics \textbf{11} (1998) 644-654.

\bibitem {cst:structures}Y. Caro, A. Seb\H{o}, M. Tarsi, \emph{Recognizing
greedy structures}, Journal of Algorithms \textbf{20} (1996) 137-156.

\bibitem {cs:note}V. Chvatal, P. J. Slater, \emph{A note on well-covered
graphs}, Quo vadis, Graph Theory Annals of Discrete Mathematics \textbf{55},
North Holland, Amsterdam (1993) 179-182.

\bibitem {de:efficienteq}M. Demange, T. Ekim, \emph{Efficient recognition of
equimatchable graphs}, Information Processing Letters \textbf{114} (2014) 66-71.

\bibitem {faenza2011}Y. Faenza, G. Oriolo, G. Stauffer, \emph{An algorithmic
decomposition of claw-free graphs leading to an }$O(n^{3})$\emph{-algorithm
for the weighted stable set problem}, SODA '2011 Proceedings of the
Twenty-Second Annual ACM-SIAM Symposium on Discrete Algorithms (2011) 630-646.

\bibitem {favaron:verywell}O. Favaron, \emph{Very well covered graphs},
Discrete Mathematics \textbf{42} (1982) 177-187.

\bibitem {fhn:wcg5}A. Finbow, B. Hartnell, R. Nowakowski, \emph{A
characterization of well-covered graphs of girth 5 or greater}, Journal of
Combinatorial Theory B \textbf{57} (1993) 44-68.

\bibitem {fhn:wc45}A. Finbow, B. Hartnell, R. Nowakowski, \emph{A
characterization of well-covered graphs that contain neither 4- nor 5-cycles},
Journal of Graph Theory \textbf{18} (1994) 713-721.

\bibitem {gabow:wmatching}H. Gabow, \emph{Data structures for weighted
matching and nearest common ancestors with linking}, SODA '90 Proceedings of
the first annual ACM-SIAM symposium on Discrete algorithms (1990) 434-443.

\bibitem {lpp:equimatchable}M. Lesk, M. D. Plummer, W. R. Pulleyblank,
\emph{Equimatchable graphs}, Graphs Theory and Combinatorics, B. Bollobas ed.,
Academic press, London, (1984) 239-254.

\bibitem {lt:relating}V. E. Levit, D. Tankus, \emph{On relating edges in
well-covered graphs without cycles of length 4 and 6}, Graph Theory,
Computational Intelligence and Thought: Essays Dedicated to Martin Charles
Golumbic on the Occasion of His 60th Birthday, Lecture Notes in Computer
Science \textbf{5420} (2009) 144-147.

\bibitem {lt:wc4567}V. E. Levit, D. Tankus \emph{Weighted well-covered graphs
without $C_{4}$, $C_{5}$, $C_{6}$, $C_{7}$}, Discrete Applied Mathematics
\textbf{159} (2011) 354-359.

\bibitem {lt:wwc456}V. E. Levit, D. Tankus, \emph{Weighted well-covered graphs
without cycles of lengths 4, 5 and 6}, arXiv:1210.6918 [cs.DM],
http://arxiv.org/pdf/1210.6918v1.pdf (available online).

\bibitem {lt:relatedc4}V. E. Levit, D. Tankus, \emph{On relating edges in
graphs without cycles of length 4}, Journal of Discrete Algorithms (2013),
http://dx.doi.org/10.1016/j.jda.2013.09.007, (available online).

\bibitem {minty:wclaw}G. J. Minty, \emph{On maximal independent sets of
vertices in claw-free graphs}, Journal of Combinatorial Theory B \textbf{28}
(1980) 284-304.

\bibitem {plummer:survey}M. D. Plummer, \emph{Well-covered graphs: a survey},
Quaestiones Mathematics \textbf{16} (1993) 253-287.

\bibitem {ptv:chordal}E. Prisner, J. Topp and P. D. Vestergaard,
\emph{Well-covered simplicial, chordal and circular arc graphs}, Journal of
Graph Theory \textbf{21} (1996), 113-119.

\bibitem {ravindra:well-covered}G. Ravindra, \emph{Well-covered graphs},
Journal of Combinatorics, Information and System Sciences \textbf{2} (1977) 20-21.

\bibitem {sknryn:compwc}R. S. Sankaranarayana, L. K. Stewart, \emph{Complexity
results for well-covered graphs}, Networks \textbf{22} (1992), 247-262.

\bibitem {tata:wck13f}D. Tankus, M. Tarsi, \emph{Well-covered claw-free
graphs}, Journal of Combinatorial Theory B \textbf{66} (1996) 293-302.

\bibitem {tata:wck13fn}D. Tankus, M. Tarsi, \emph{The structure of
well-covered graphs and the complexity of their recognition problems}, Journal
of Combinatorial Theory B \textbf{69} (1997) 230-233.

\bibitem {tata:hamilton}D. Tankus, M. Tarsi, \emph{Greedily constructing
Hamiltonian paths, Hamiltonian cycles and maximum linear forests}, Discrete
Mathematics \textbf{307} (2007) 1833-1843.
\end{thebibliography}
\end{document}